\documentclass[11pt]{article}
\usepackage{fullpage,verbatim,amsmath,amssymb,amsthm,float,graphicx,setspace}

\floatstyle{ruled}
\newfloat{algoritcitathm}{thp}{lop}
\floatname{algorithm}{Algorithm}

\usepackage{amsmath,wasysym}
\usepackage{amsfonts}
\usepackage{fullpage}
\usepackage{epsfig}
\usepackage{array}
\usepackage{epsfig}
\usepackage{amsmath}
\usepackage{amssymb}
\usepackage{graphicx}
\usepackage{subfigure}
\usepackage{color}
\usepackage{algorithm}
\usepackage{algorithmic}
\usepackage{url,array,enumitem}

\newtheorem{theorem}{Theorem}
\newtheorem{property}{Property}
\newtheorem{lemma}{Lemma}
\newtheorem{remark}{Remark}

\newcommand{\bad}{\tau}
\newcommand{\randNum}{{\tt randNum}}
\newcommand{\randCl}{{\tt randCl}}
\newcommand{\exchange}{{\tt exchange}}
\newcommand{\G}{\ensuremath{\widehat G}}

\newcommand{\leave}{{\tt Remove}}
\newcommand{\join}{{\tt Add}}
\newcommand{\link}{{\tt Link}}
\newcommand{\CTRW}{{\tt CTRW}}
\newcommand{\N}{{\def\N{\mbox{I\hspace{-.15em}N}}}}

\begin{document}
\title{Highly Dynamic Distributed Computing with Byzantine Failures}

\author{Rachid Guerraoui\protect\footnote{Email: rachid.guerraoui@epfl.ch; Tel: +41 21 693 5272}\\EPFL, Switzerland \and 
Florian Huc\protect\footnote{Email: florian.huc.@epfl.ch; Tel: +41 21 693 8125
}
\\EPFL, Switzerland \and Anne-Marie Kermarrec\protect\footnote{Email: anne-marie.kermarrec.@inria.fr; Tel: +33 2 99 84 25 98}\\ INRIA Rennes Bretagne-Atlantique, France}

\date{}


\maketitle

\begin{abstract}
This paper shows for the first time that  distributed computing can be both  reliable and efficient in an environment that is both highly dynamic and hostile. 
More specifically, we show how to maintain clusters of size $O(\log N)$, each containing more than two thirds of honest nodes with high probability,  within a system whose size  can vary \textit{polynomially} with respect to  its initial size. Furthermore, the communication cost induced by each node arrival or departure is polylogarithmic with respect to $N$, the maximal size of the system. 
Our 
clustering  can be achieved despite the presence of a  Byzantine adversary controlling a fraction $\bad \leq \frac{1}{3}-\epsilon$ of the nodes, for some fixed constant $\epsilon > 0$, independent of $N$.
So far, such a clustering could only be performed for systems who size can vary constantly and it was not clear whether that was at all possible for polynomial variances.
\end{abstract}

\small

\hspace{0.5cm}\textbf{Keywords:} Byzantine failures, random walks, dynamic networks
\normalsize

\newpage


\section{Introduction}

Distributed computing can be achieved reliably in a system where at most one third of the processes are controlled by an adversary. 
Typically, assuming some synchrony, the seminal agreement problem ~\cite{Lamport1982} can be solved and used to emulate a single highly available process. This is a basic building block to achieve distributed computations in a reliable manner.  Yet, with a large number of nodes, this technique is very expensive.  One way to reduce the complexity consists in clustering the nodes within smaller subsets, picked randomly, so  that each cluster contains two third of correct nodes whp, e.g., as proposed in \cite{Galil198781}. In short, instead of reducing a system of many processes into a system of one reliable process that performs the computation, the idea here is to reduce it to a system of several reliable processes, each corresponding to one of the clusters. These processes share the load of the computations reducing thereby their complexity.

So far, clustering techniques mainly assumed a static distributed system: the number $n$ of processes is assumed to be fixed a priori and processes do not join or leave the system \cite{giurgiu2010computing} (a few can typically fail). Some approaches have explored dynamic settings, but in a limited fashion: the number of processes $n$ is assumed to only vary by a constant factor \cite{scheideler,Awerbuch2009,SRDS12,Scheideler2005}. 
Yet, whether this is  at all possible to go beyond has been considered an open question so far \cite{King2011,King2010}.

This paper answers the question positively. We show, for the first time, that it is possible to perform distributed computing reliably and efficiently in a system which size can vary in a \textit{polynomial} manner. At the heart of this result  lies a new technique to partition nodes in a dynamic number of clusters, which involves a radical departure from previous schemes that assume a static number of clusters~\cite{scheideler,Awerbuch2009,SRDS12,Scheideler2005}. Indeed, tolerating an increase in the number of nodes from $n$ to $n^2$ (and more generally from $n^{1/y}$ to $n^z$ for some constants $y,z>1$), with a static number of clusters, yields a significant increase in the  number of nodes within each cluster, leading to a high-complexity computation, in the vein of a single cluster approach. However, handling dynamic clusters is not trivial. For instance, using classical De Brujin graphs for clustering~\cite{scheideler} in a dynamic setting requires a good estimation of the number of nodes. In turn, this potentially  requires techniques with high complexity, e.g., typically $\tilde O(n^{3/2})$~\cite{Katzir11}.

Our clustering approach achieves a polylogarithmic complexity by using random walks on expander graphs with small degrees. 
To ensure that each cluster contains two thirds of correct nodes with high probability, we exchange nodes between
 clusters whenever new nodes join or leave the system. The nodes that are candidate to the exchange are selected using continuous random walks~\cite{BookMarkovProcesses}. These provide a uniformly chosen sample even if the underlying graph is not regular. To ensure that a walk ends up fast on a node picked quasi uniformly, we connect clusters through small degree expanders~\cite{krebs2011expander}. 

The distributed construction of this expander requires specific care in regulating the choice of edges. Although several expanders could be used, our approach relies on OVER, a technique (Over-Valued Erd\"os R\`einy graph) from Erd\"os R\`einy random graphs to preserve a small degree and a good expansion coefficient. OVER is described in the Appendix for space reasons. This technique tolerates more crashes than \cite{Aspnes08,amin2004,law2003} and yields a different degree than \cite{AspnesW05}.   In the rest of the paper, we present NOW (Neighbors On Watch),  a protocol  maintaining the cluster partition  in Section~\ref{overview} and analyze it in Section~\ref{analysis}.  We review the related work in Section~\ref{related} and conclude in  
Section~\ref{conclusion}. 
Some background about continuous time random walk as well as more details about  possible applications of our clustering technique are also provided in the appendix.

\section{Model and background}\label{model}
\paragraph{System assumptions.} In short, our network model is the one of \cite{Awerbuch2009} with the difference that we allow the size of the system to vary \textit{polynomially}. 
More specifically, we consider a dynamic synchronous network with a discrete time variable $t_i$. Each node can send messages to any node it knows through a private channel; in this sense the network is reconfigurable as connections between nodes can be added or removed. We do not assume that each node knows all other nodes in the network (except during the initialization phase in which the global structure of the network is computed once). Instead, each node knows $\mathrm{polylog}(N)$ nodes and only knows an upper bound on the current size of the network. We also assume that, initially, the number of nodes is $n_{t_0}$ for some  $\sqrt{N} \leq n_{t_0} \leq N$ , and the current number of nodes $n$ in the network remains between $\sqrt{N}$ and $N$ (this can be relaxed to $N^{1/y}\leq n \leq N^{z}$ for all constants $y,z>1$). The size of the network can increase or decrease at any time. For simplicity of presentation, we  assume (as in~\cite{Awerbuch2009,law2003}) 
that when a node joins or leaves, the actions relative to previous joins and leaves are over. 
This corresponds to a time step.˜\footnote{However, the analysis can be generalized to several parallel join and leave operations.} Moreover, nodes do not need to take any specific action when leaving the network. Instead, we assume a mechanism enabling a node to detect if one of its neighbors has crashed or left the network.


\paragraph{Adversary model.} 
Our adversary is that of \cite{Awerbuch2009} with the difference that in our case it   controls a fraction $\tau \leq \frac{1}{3}-\epsilon$ (for some constant $\epsilon >0$) nodes, from the beginning ({\it{vs.}} $\tau \leq \frac{1}{2}-\epsilon$ after some initialization phase; note that using cryptographic tools, we could also assume $\tau \leq \frac{1}{2}-\epsilon$ by leveraging broadcast algorithms~\cite{garay2011adaptively}.).

NOW tolerates a static Byzantine (sometimes called \emph{active}) adversary controlling a fraction $\tau \leq \frac{1}{3}-\epsilon$ (for some constant $\epsilon >0$) of the nodes, having a full knowledge of the network at any time, as in \cite{Awerbuch2009,King2011}, i.e it knows the position of any node at any time. A typical objective for the adversary is to gain the lead in one (or more) of the clusters. At the beginning of the protocol, the adversary can choose a fraction $\tau$ of the nodes to corrupt. We assume that, at initialization, the honest nodes form a connected component, that the adversary cannot split it into disjoint parts, that each node controlled by the adversary is adjacent to at least one honest node, and that no honest node leaves or joins the network until the initialization is over. Also, nodes' identities cannot be forged. Moreover during the execution of the protocol, each time a node joins the network, the adversary can choose to corrupt it or not, as in \cite{scheideler,Awerbuch2009}. However, the adversary cannot decide to corrupt nodes at a later time (in this sense the adversary is static and not adaptive).  Furthermore, the adversary can induce churn as in \cite{scheideler,Awerbuch2009} by join-leave attacks or by forcing honest nodes to leave the system (e.g., through a DOS attack). The size of the network can vary polynomially and each node is assigned a unique identifier.

\paragraph{Background on OVER: expander graph.}

Our clustering technique, which we call NOW (Neighbors On Watch) and that maintains a cluster partition is based on a protocol to distributely maintain an expander overlay.
 Although  various expanders (e.g. \cite{AspnesW05}) could be used,  we assume that  NOW relies on OVER. For space reasons, the detailed description of OVER is deferred to the appendix (Section~\ref{expander}).  In OVER,  the graph vertices represent the clusters of nodes maintained by NOW, hence they can be considered as honest since each vertex is composed of more than two thirds of honest nodes whp. We further assume that each vertex leaving the overlay graph is chosen at random (this assumption will be ensured in Section~\ref{maintenance}). 

OVER ensures  that, starting from a random graph drawn from the Erd\"os-R\'enyi model, whp, at any time during a sequence of vertex additions and removals polynomial in $N$, the resulting graph has a large isoperimetric constant and a low degree (ensuring properties \ref{p(G)} and \ref{pMaxDegree}).  We  use the notation $\G^R$ where the $\hat\ $ relates to the fact that we consider an overlay, and the $\ ^R$ that it is an instance of a random graph. The evolution of the graph is represented by a sequence $\G^R_{t_0},\ldots, \G^R_{t_i}, \ldots$. $n_{t_i}$ denotes the number of vertices of $\G^R_{t_i}$. 

\begin{property}\label{p(G)} 
Whp, at any time $t$ after a number of time steps polynomial in $n$, $\G^R_{t}=(\hat V^R_{t},\hat E^R_{t})$ has an isoperimetric constant $I(\G^R_{t}) \geq \log^{1+\alpha} N/2$, where:

 $I(\G^R_{t})= \inf_{S\subset \hat V^R_{t}:|S|\leq n_t/2}E(S,\bar S)/|S|$.
\end{property}
\begin{property}[Maximum degree of $\G^R$] \label{pMaxDegree}
Whp, at any time $t$ after a number of time steps polynomial in $n$, $\G^R_{t}$ has maximal degree at most $c\log^{1+\alpha} N$ for a large enough constant $c$ and an arbitrarily small (pre-)chosen constant $\alpha$. 
\end{property}

Those properties enable to achieve short random walks leading to pick nodes uniformly at random. 
Note that OVER enables to tolerate simultaneous failures as long as the targets are picked uniformly at random. NOW together with OVER can also tolerate the failures of nodes chosen by the adversary as long as one failure per round is assumed.

\paragraph{Notations.} We use the time step as a subscript to indicate the instant at which a variable is considered (e.g., $n_{t_i}$ is the number of nodes at time $t_i$, $\#C_{t_i}$ is the number of clusters, and $|C_j|_{t_i}$ the size of $C_j$). We may omit the index of the time step when there is no ambiguity (e.g., $n$ stands for the current number of nodes in the network).  
The \emph{communication cost} is the number of messages\footnote{We consider messages of identical size. Hence the communication cost is proportional to the number of bits sent.} exchanged, and the \emph{round complexity}, is the number of communication rounds (i.e. the number of successive messages) required by a protocol to terminate. Notice that a time step is composed of several communication rounds, but we will prove that they are $\mathrm{polylog}(N)$. Given a graph $G=(V,E)$, and a vertex $v\in G$, we denote by $d_v$ its degree. Similarly, for a given cluster $C$, $d_C$ denotes the number of clusters adjacent to $C$. 



\section{NOW: Overlay of clusters}\label{overview}

NOW (Neighbors On Watch) maintains both an overlay of clusters and the  partition of the nodes into clusters. NOW relies on the fact that the overlay is  guaranteed to have a low maximum degree and good expansion properties. This is provided by the protocol OVER that we present in the appendix but could also be 
ensured by other protocols which differ either in the number of failures they can provide \cite{Aspnes08,amin2004,law2003}
 or their degree (e.g., 4 in \cite{AspnesW05} instead of $\log^{1+\alpha} N$ in OVER for some arbitrarily small constant $\alpha >0$) ).  NOW further ensures that each cluster contains more than two thirds of honest nodes whp. The clusters have size $O(\log N)$ and are used to inhibit the behavior of the Byzantine nodes. NOW relies on  two phases:  \textit{initialization} and  \textit{maintenance}. In a nutshell, the initialization phase generates the initial overlay, while the maintenance phase ensures that after a polynomially long sequence of leave and join operations, the required properties still hold. The overlay $\G^R$ is first constructed during the initialization phase of NOW, and recursively maintained by OVER as described in Appendix, Section~\ref{expander}.

\subsection{Preliminaries.}
A node of a  cluster $C$  is linked to all the other nodes of $C$ and knows their identities.
An edge between two clusters $C_i$ and $C_j$ in $\G^R$ means that all nodes of $C_i$ are linked to all  nodes of $C_j$ and know their identities (and \emph{vice-versa}). A node only needs to know the identities of the nodes in its cluster and the neighboring ones. The initialization phase (Section \ref{initialisation}) is itself divided into two sub-phases. First, a discovery algorithm is run in order for the nodes to acquire a global knowledge of the network. Afterwards, a Byzantine agreement algorithm~\cite{King2010} is used to construct an initial overlay of clusters. The maintenance phase ensures that each cluster contains more than two thirds of honest nodes whp when nodes join or leave and preserves the properties of the overlay.


\begin{figure*}[htb]
  \begin{center}
\scalebox{0.6}{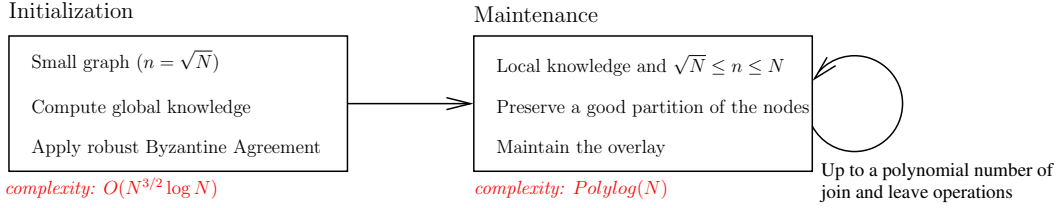}
\caption{Overview of NOW.}\vspace{-0.6cm}
   \end{center}
\end{figure*}

\paragraph{Random number generation.}
We assume the existence of \randNum, a  distributed random number generation protocol, enabling the nodes of a cluster to agree on a common integer chosen uniformly at random from the interval $(0,r)$. \texttt{randNum} is secure as long as the Byzantine nodes are less than two thirds in the cluster and is presented in the appendix (Section \ref{randnum}).

\paragraph{Cluster random choice.}
Furthermore, we assume the existence of a function called {\randCl} (Algorithm \ref{algo_id_randCl} in Appendix), to randomly select a cluster. To achieve the random selection
 (\randCl),  we perform a biased CTRW~\footnote{A vertex $C_i$ of $\G^R$ is a cluster in $G$. A biased CTRW from $C_i$ is a sequence of CTRW as follows: the nodes of $C_i$ choose collaboratively the next cluster $C_j$ and decrease the duration of the CTRW using {\randNum} which goes on similarly. When the remaining duration is negative or null, a random number between 0 and 1 is chosen. If it is smaller than $|C_i|/\max_{C} |C|$, the biased CTRW ends, otherwise a CTRW starts again. A node of a cluster $C_j$ pursues the random walk if and only if it receives an identical message from at least half plus one of the nodes of the neighboring cluster from which the CTRW comes.} on $\G^R$, the overlay.  We bias our CTRW such that  a cluster is chosen  according to the distribution $(|C_i|/n)$. With clusters of size $O(\log N)$, this primitive has an expected communication cost of $O(\log^5 N)$. Indeed, the expected number of clusters visited during the walk is $O(\log^3 N)$ (whp, we do $O(\log n)$ CTRW each of length $O(\log^2 n)$) and at each cluster a random integer from the range $(0,O(\log^{1+\alpha} N))$ is generated at a cost of $O(\log^2 N)$. The expected round complexity of this primitive is $O(\log^4 N)$. 

\paragraph{Node shuffling.} In order to avoid an adversary to focus on one cluster and gradually pollute it with Byzantine nodes, shuffling nodes between clusters 
is  necessary upon nodes arrival and departure. The shuffling is implemented by the algorithm called {\exchange} and detailed in Appendix (Algorithm \ref{algo_id_exchange}).
Basically  some clusters exchange their nodes with nodes chosen at random from other clusters. For each node $x$ to be exchanged from cluster $C$ ($x$ is determined by the protocol {\exchange}), a cluster is chosen at random using \randCl . The chosen cluster, $C'$, is informed that it will receive $x$. The cluster $C'$ chooses one of its nodes (using the primitive \randNum) to send in replacement of $x$. 
During an exchange, if $C$ is adjacent to another cluster, the nodes of this cluster are informed of the  new composition of $C$. This step is fundamental since a node from a neighboring cluster accepts a message from $C$ if and only if at least half plus one of the nodes of $C$ send it. The new nodes of $C$ are informed by the former nodes of this cluster of the local structure of the overlay (i.e., the neighboring clusters of $C$ in the overlay). The expected communication cost and round complexity of {\exchange} are $O(\log^6 N)$ and $O(\log^4 N)$.

\subsection{NOW: Initialization Phase}
\label{initialisation}

\paragraph{Network Discovery.}
The protocol starts by running an algorithm that informs each node of the identifiers of all other nodes. The global knowledge of the nodes in the network is needed only at initialization. Note that this computation is performed while the size of the network is still ``small'' in practice. Afterwards, it is possible to use standard off-the-shelf Byzantine agreement algorithms to construct an initial partition forming the vertices of the overlay~$\G^R$.  This algorithm (Algorithm~\ref{algo_id_distribution} provided in Appendix) terminates after a number of communication rounds at most the diameter of the graph considering only the edges adjacent to at least one honest node. When the algorithm terminates, it is guaranteed that all honest nodes know the identities of all nodes in the network. Its communication cost is $O(n\times e)$ where $e=|E|$ (see Appendix for the theorem and details).

\paragraph{Clusterization.}
Once all the honest nodes know the identities of all the nodes in the network, any Byzantine agreement protocol can be used, such as \cite{King2010} whose complexity is $\tilde O(n\sqrt{n})$. This protocol works in the presence of a static Byzantine adversary controlling less than $1/3-\epsilon$ of the nodes for some positive constant $\epsilon$.
A representative cluster of logarithmic size containing more than two thirds of honest nodes is selected. Afterwards, we use the nodes of this representative cluster to randomly partition the network into $\#C$ clusters, $\{C_1,\dots,C_{\#C}\}$, each of size $k\log N$, for some constant $k$. The constant $k$ is a security parameter of the protocol that is chosen \emph{a priori} depending on the requirements of the application considered: the higher $k$, the less chances the adversary has to control more than a third of the nodes of one of the clusters. Choosing the partition at random ensures that whp, there is more than two thirds of honest nodes in each cluster. This can be proved using standard Chernoff bound and union bound arguments. To obtain a random partition, it is sufficient for the representative cluster to order the nodes at random by calling the primitive {\randNum}. Once the random ordering has been computed, the partition is obtained by taking for each cluster $k\log N$ successive nodes.  Afterwards, $\G^R_{t_0}$ is initiated on top of this partition: for each pair of clusters, the representative cluster determines with probability $p=\log^{1+\alpha} N/\sqrt{N}$ whether or not they will be linked by an edge in $\G^R_{t_0}$. Finally, the representative cluster tells each node $x$ the cluster it belongs to, the identities of the other nodes in this cluster, and the adjacent clusters as well as their composition (i.e., the identities of the nodes). The node $x$ is ``linked'' to all these nodes and can for efficiency purposes forget the identifier of any other node that it may know. It is fundamental for the security of our protocol that each cluster contains more than two thirds of honest nodes. Indeed, a node receiving a message from all the nodes of a particular cluster considers this message valid if and only if, it receives the same message from more than half of the nodes of this cluster. Using this rule for inter-cluster communication, together with the condition that each cluster has more than two thirds of honest nodes, is sufficient to ensure the correctness of the protocol. 

\subsection{NOW: Maintenance Phase}\label{maintenance}

\begin{figure*}[tb]
  \begin{center}
\scalebox{0.42}{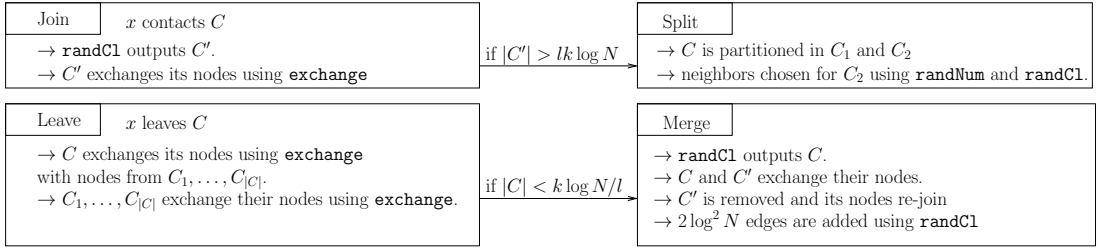}
\caption{Maintenance of the overlay. Each operation has a $\mathrm{polylog}(N)$ complexity.}
   \end{center}
\end{figure*}

While the initialization phase of NOW ensures the desired properties for both the overlay and the clusters, maintaining these properties under high dynamics is challenging. In this section, we describe how to preserve the property that each cluster is composed of an honest majority in the presence of nodes join and leave operations.  
Shuffling the network is crucial  at this point as mentioned in \cite{scheideler,Awerbuch2009,Scheideler2005} to avoid the adversary to control a majority of nodes in a cluster after a few steps by using a very simple strategy: the adversary chooses a specific cluster and keeps adding and removing the Byzantine nodes until they fall into that cluster. Similarly, it is crucial to introduce dynamics with shuffling if nodes are forced to leave the network by the adversary. The shuffling is generated upon \textit{Join} and \textit{Leave} operations. Complementary, the \textit{Split} and \textit{Merge} operations ensure that the clusters remain of size $\Omega(\log N)$, and that the required properties of $\G^R$ (i.e., expansion and low maximum degree) are preserved. 

The NOW following operations are invoked by the nodes upon joining, or leaving the network, or simultaneously by all the nodes of a cluster involved in a split or a merge operation.

\paragraph{Join.}
This operation (as well as the leave operation), initiated by a node joining the network, is inspired by \cite{scheideler,Awerbuch2009,Scheideler2005}. When a node $x$ joins the network, we assume that it gets in contact with a cluster of the overlay. This cluster chooses another cluster using {\randCl} in which $x$ is inserted. The chosen cluster proceeds by inserting $x$ and uses {\exchange} for all of its nodes. This operation has a communication cost of $\mathrm{polylog}(N)$.

\begin{algorithm}
\caption{Join operation.}
\label{algo_id_join}
\footnotesize
\begin{algorithmic}
\REQUIRE Node $x$ contacting cluster $C$ to join the network.
\ENSURE The preservation of the properties of the overlay and of the clusters.
\STATE Nodes of $C$ choose a cluster $C'$ using {\randCl}.
\STATE All nodes of $C'$ add $x$ to their local view of $C'$.
\STATE All nodes of $C'$ send a message to all the nodes from the neighboring clusters informing that $x$ is added to $C'$.
\STATE All nodes of $C'$ send their neighborhood to $x$ using the path used to find $C'$ in {\randCl}.
\IF{$|C'| > kl\log n$}
\STATE Nodes of $C'$ compute a partition of $C'$ into two parts of roughly the same size using {\randCl}: $C_1$ and $C_2$.
\STATE Nodes of $C_1$ keep their neighborhood.
\STATE Nodes of $C_1$ and $C_2$ send a message informing that $C'$ is replaced by $C_1$ to the neighbors of $C_1$.
\STATE Nodes of $C_2$ are given a new neighborhood using $\join(C_2)$ (Algorithm~\ref{algo_id_add} of OVER).
\ENDIF
\end{algorithmic}
\end{algorithm}

\paragraph{Split.} This operation is initiated simultaneously by all  nodes of a cluster $C$ if after a join operation, the size of this cluster is larger than $lk\log N$ for some fixed parameter $l$ ($l$ is a constant greater than $\sqrt 2$ which influences the number of split and merge operations). Then $C$ has to be split in two, the old and the new clusters. To this end, the nodes of $C$ generate a random partition of $C$. The old cluster keeps its neighbors in $\G^R$, whereas the new cluster is added to the overlay using {\join} as described in Section \ref{expander}.
This procedure has a communication cost of $\mathrm{polylog}(N)$ and a $O(\log^4 N)$ round complexity. Recall that each node knows the exact composition of its cluster, therefore a split operation can be easily achieved.

\paragraph{Leave.}  This operation occurs when a node from a cluster $C$ leaves the network or when the other nodes of $C$ detect its absence. $C$ exchanges all its nodes using the primitive \exchange. Then, a cluster receiving one or more nodes from $C$ execute {\exchange} for all of its nodes. This process has a communication cost of $\mathrm{polylog}(N)$ and a $O(\log^4 N)$ round complexity.

\begin{algorithm}
\caption{Leave operation.}
\label{algo_id_leave}
\footnotesize
\begin{algorithmic}
\REQUIRE Node $x$ from a cluster $C$ leaving the network.
\ENSURE The preservation of the properties of the overlay and the clusters.
\STATE Nodes of $C$ remove $x$ from their view.
\STATE Nodes of $C$ send a message to their neighbors informing them to remove $x$ from their view.
\STATE A node that is a neighbor of $C$ receiving a message to remove $x\in C$ from more than half of the nodes of $C$ removes it from its view.
\STATE $C$ exchanges its nodes using {\exchange}.
\STATE A cluster exchanging one or more of its nodes with $C$ execute the {\exchange} procedure.
\IF{$|C'| < k\log n/l$}
\STATE Nodes of $C$ inform all their neighbors that $C$ is removed.
\STATE Nodes of $C$ execute $\leave(C_1)$ ~(Algorithm \ref{algo_id_remove}) of OVER.
\STATE A node that is a neighbor of $C$ receiving a message that $C$ is removed from more than half of the nodes of $C$ removes it from its view.
\STATE Nodes of $C$ execute Algorithm \ref{algo_id_join} as to rejoin the network.
\ENDIF
\end{algorithmic}
\end{algorithm}

\paragraph{Merge.} 
This operation is initiated simultaneously by all nodes of a cluster $C$ containing less than $\frac{k\log N}{l}$ users (for the same fixed parameter $l$ described previously).
 In this situation, a cluster, chosen at random in order to ensure Properties \ref{p(G)} and \ref{pMaxDegree}, has to be removed. This is achieved using the primitive {\randCl}. 
 Nodes in $C$ proceed as if they were joining the network while the nodes from the chosen cluster $C'$ become members of $C$. In $\G^R$, $C'$ is removed by using the operation {\leave} described in Section \ref{expander}.

\section{NOW: Analysis}
\label{analysis}

In this section, we prove that after a polynomial sequence of join and leave operations (some of them inducing some splitting and merging of clusters), each cluster contains more than two thirds of honest nodes as long as the fraction of Byzantine nodes $\bad$ controlled by the adversary is smaller than $1/3-\epsilon$ (for some constant $\epsilon>0$ independent of $n$). 

The results are proved under the assumption that the
random choices of nodes are perfectly uniform (i.e, the small bias induced
by the random walk is ignored). This assumption is justifed by the fact that we consider a mixing
time after which the distance from the desired distribution is $O(n^{-c})$ 
for some arbitrarily large constant $c$.
More specifically, we describe the output of a CTRW using two random variables $X$ and $Y$. $X$ indicates whether or not the output of the CTRW has the desired distribution and is defined as follows: we consider the probability distribution $\cal D$ of the endpoints of a CTRW, and set $p_v$ as the probability node $v$ is hit.
Set $p_{min} = \min_v (p_v)$. The binary random variable $X$ has value 1 with probability $n\times p_{min}$ and 0 otherwise. $Y$ is equal to node $v$ with probability $(p_v-p_{min})/(\sum_w (p_w-p_{min}))$. We can reproduce $\cal D$ by first evaluating $X$.
Then, if $X=1$, the endpoint is picked   according to the desired distribution. Else, the endpoint is picked according to $Y$. We have $P(X=0) \leq n \times \max(p_v-p_{min}) = O(n^{-c+1})$, which means that the probability of the endpoint not to be picked as desired is $O(n^{-c+1})$. Conditional to that, in the following we assume that the random choices made using a CTRW are as desired, i.e ($|C|/n$) for each cluster $C$ where $|C|$ is its size.

\subsection{Status of a cluster after \exchange}

At each time step, we assume that either a join or leave operation takes place or nothing occurs.  
These operations may in turn induce the splitting or merging of clusters. A split operation is done directly at the time it occurs, whereas, when a cluster executes a merge operation, we consider that its nodes re-join the network in subsequent time steps inducing normal join operations. Given a cluster $C$, $p^C_{t}$ is the proportion of Byzantine nodes in $C$ at time~$t$. 
\setcounter{lemma}{0}
\begin{lemma}[2/3 of honest nodes in a cluster]\label{afterExchange}
If a cluster $C$ has exchanged all its nodes at time step $t$, we have $P(p^C_{t}>\tau(1+\epsilon))\leq n^{-\gamma}$, for any positive constant $\gamma$, as long as the security parameter $k$ is large enough.
\end{lemma}
\begin{proof}
When a cluster $C$ exchanges one of its nodes with another cluster, this cluster is first selected at random according to the probability distribution $(|C|_{t_i}/n)$, and then a node is chosen out of it uniformly at random. In this scenario, the probability of performing an exchange with a Byzantine node is~$\tau$. 

Using standard Chernoff bound arguments, we can derive the following result on the number $X$ of Byzantine nodes among $|C|_{t_i}$ nodes: $P(X>(1+\epsilon)\tau|C|_{t_i})\leq e^{-\epsilon^2\tau |C|_{t_i}/3}$. Therefore as $|C|_{t_i} \geq (k\log N)/l$, we have  $P(X>(1+\epsilon)\tau|C|_{t_i})\leq N^{-\gamma}$ when $k$ is sufficiently large for some constant~$\gamma$.
\end{proof}

This lemma is a consequence of the Chernoff bound arguments \cite{ProbabilisticMethod} and implies that to obtain more than two thirds of honest nodes in a cluster whp, it is sufficient that $\tau+\epsilon < 1/3$, which is true by assumption on~$\bad$. 

\begin{remark}[Increasing the robustness] One can tolerate a fraction of Byzantine nodes up to $1/2-\epsilon$, but then we need to use cryptographic tools to allow for broadcast and Byzantine agreement.\end{remark}

\subsection{Evolution of the divergence}

To summarize, we have seen that each time a cluster exchanges all of its nodes, as long as $\tau(1+\epsilon)<1/3$, we obtain more than two thirds of honest nodes whp in the resulting cluster. We now proceed by proving that in between two exchanges, this property also holds. To realize this, we focus on a specific cluster $C$ and consider a sequence of $s$ join and leave operations. 

We first prove that if the cluster has less than a $\tau(1+\epsilon/2)$ fraction of Byzantine nodes, then after it has exchanged $O(\log N)$ of its nodes, it does not have more than a $\tau(1+\epsilon)$ fraction of Byzantine nodes. Then, we prove that if it has between a $\tau(1+\epsilon/2)$ and  $\tau(1+\epsilon)$ fraction of Byzantine nodes, then after it has exchanged $O(\log N)$ of its nodes, it has less than a $\tau(1+\epsilon/2)$ fraction of Byzantine nodes whp.

\begin{lemma}\label{1}
If a cluster $C$ has less than $\tau(1+\epsilon/2)|C|$ Byzantine nodes, then after  $O(\log N)$ node exchanges with nodes chosen uniformly at random, the cluster does not contain more than $\tau(1+\epsilon)|C|$ Byzantine nodes whp.
\end{lemma}
\begin{proof}
A cluster $C$ with a fraction $p$ of Byzantine nodes has a probability at most $p(1-\tau)$ to have this fraction decreased by $1/|C|$, and at least $(1-p)\tau$ to have it increased by the same amount. If this fraction is at most $\tau(1+\epsilon/2)$, we prove that it increases by $\epsilon$ with probability $o(1/N^\gamma)$, for $\gamma$ being arbitrarily large depending on the chosen value of~$k$. 

The fraction of Byzantine nodes in the cluster is dominated by the martingale with starting state $\tau(1+\epsilon/2)$, which increases or decreases by $1/|C|$ with probability $\tau$. We now show that whp, this martingale will not exceed $\tau(1+\epsilon)$ after $O(\log N)$ steps (recall that $k\log N/l\leq |C| \leq kl\log N$). 

For $k$ large enough, let $T^{exchange}$ stands for the number of exchanges. It is $O(\log N)$ and hence there is a constant $M$ such that $T\leq M\log N$. We can derived from Azuma-Hoeffding's inequality that:
\begin{align*}
Prob(p^{C}>\tau(1+\epsilon/2)) & <e^{-\epsilon^{2}/4\sum_{i=1}^{T^{exchange}}1/|C|^2} \\
&\leq e^{-\epsilon(k/l)^2 \log^2 N/4(M\log N)}\\
&=e^{-\epsilon(k/l)^2 \log(N)/4M}=n^{-\gamma}
\end{align*}

\end{proof}

Similarly, if a cluster has more than a $\tau(1+\epsilon/2)$ fraction of Byzantine nodes, we have that after $O(\log N)$ exchanges, the cluster has less than a $\tau(1+\epsilon/2)$ fraction of Byzantine nodes.

\begin{lemma}\label{2}
Given a cluster $C$ whose fraction of Byzantine nodes is between $\tau(1+\epsilon)$ and $\tau(1+\epsilon/2)$ (for some constant $\epsilon > 0$ independent of $n$), then whp, the fraction of Byzantine nodes in this cluster is less than $\tau(1+\epsilon/2)$ after $O(\log N)$ exchanges with nodes chosen uniformly at random. 
\end{lemma}
\begin{proof}

We use the same arguments for the previous theorem. Here, the fraction of Byzantine node will decrease of $1/|C|$ with probability at least $\tau(1+\epsilon/2)$ and will increase by $1/|C|$ with probability~$\tau$. Therefore, as we start from a fraction of at most $\tau(1+\epsilon)$, whp, after $O(\log N)$ exchanges, the fraction of Byzantine nodes in this cluster is less than $\tau(1+\epsilon/2)$. 
\end{proof}

When we look at a sequence of $s$ exchanges affecting a given cluster $C$, we can split this sequence in alternating sub-sequences to apply Lemmas \ref{1} and \ref{2}. Some sequences might lead to a fraction of Byzantine nodes between $\tau(1+\epsilon/2)$ and $\tau(1+\epsilon)$, while the following one will lead to a fraction of Byzantine nodes bellow $\tau(1+\epsilon/2)$ whp. Hence, for a sequence $s$ whose length is polynomial, by the union bound, we obtain that is there is always (whp) more than two thirds of honest node in each cluster for an adequate~$k$.

\setcounter{theorem}{2}
\begin{theorem}
\label{cor}
Whp, after a number of steps polynomial in $N$, at each time step, all clusters are composed of more than two thirds of honest nodes.
\end{theorem}
\begin{proof}
Notice that to apply the previous lemmas, one has to ensure that the exchanged nodes are replaced by nodes chosen uniformly at random. This is ensured by our join and leave operations. This is clear for a join operation by the use of a biased CTRW to select the replacement node. For a leave operation, this is also clear for the cluster $C$ from which the node leaves has its nodes exchanged with nodes selected uniformly at random. However, if we look at a cluster $C'$ with which $C$ has exchanged nodes, then the probability that $C'$ receives a Byzantine node is not necessarily $\tau$ as it is equal the proportion of Byzantine nodes in $C$. This is why we enforce $C'$ to exchange all its nodes.

Now, given a specific cluster, $C$ we consider an alternating sequence of time steps $t_1, \dots, t_i, \dots$ when the fraction of nodes controlled by the adversary in $C$ becomes larger or equal to $\tau(1+\epsilon/2)$ and when it becomes smaller.

Consider $i$ such that at $t_i$ the fraction of nodes controlled by the adversary in $C$ is less than $\tau(1+\epsilon/2)$ (this is in particular true at the beginning). Then at $t_{i+1}$, it becomes greater or equal to $\tau(1+\epsilon/2)$ and is less than $\tau(1+\epsilon)$. Lemma \ref{2} ensures that time step $t_{i+2}$ comes within $O(\log N)$ steps, and Lemma \ref{1} ensures that between $t_{i+1}$ and $t_{i+2}$, the adversary never controls more than a $\tau(1+\epsilon)$ fraction of nodes of the cluster.

By an union bound over all clusters, we have the announced result.
\end{proof}
\begin{remark}
\label{third}
Considering an adversary controlling at most a fraction $1/r - \epsilon$ of the nodes for some constant $\epsilon >0$ and $r\geq 2$ independent of $n$, it is possible to strengthen Theorem \ref{cor} to obtain that in all the clusters the adversary controls at most a fraction $1/r$ of the nodes.
\end{remark}

\section{Related Work}\label{related}

Several authors  studied the impact of dynamics on distributed computations~\cite{BaldoniBR09,BaldoniBN11} and overlay networks. In \cite{Baumann2009,kuhn2010distributed,kuhn2011}, the communication links of a dynamic network may be modified by the adversary  under some connectivity restrictions. In \cite{Augustine2011}, the authors study the scenario in which the adversary can force a large number of nodes of its choice to leave the network while other nodes naturally join the network at the same time. These join and leave operations impact the topology. Yet  the size of the network is assumed to remain constant. The authors assume furthermore that the nodes are connected via an expander graph. Depending on whether the adversary has to decide in advance the identities of the nodes to be kicked-out of the network, the authors propose almost-everywhere agreement protocols tolerating at each time step a churn of, respectively $O(n)$ and $O(\sqrt n)$. 
The two main differences with our work are that (1)~all nodes are assumed to be honest (i.e., the adversary is only external) and (2)~nodes are connected via an expander graph by assumption. In contrast, our protocol tolerates a Byzantine adversary controlling a constant fraction of the nodes of the network and dynamically maintains the expander graph.

Some protocols have been proposed to maintain P2P overlay networks. Some offer efficient routing strategies and tolerate  crashes, e.g. CAN, Pastry or Tapestry \cite{ratnasamy2001scalable,rowstron2001pastry,zhao2004tapestry}. Some are dedicated to asynchronous networks with concurrent joins and leaves \cite{li2004active}. However, none guarantees both that each node has a low degree and that the resulting overlay exhibits  good expansion properties in the sense we require here.
Protocols such as SHELL \cite{scheideler2009distributed}  organize  peers into a heap structure resilient to large Sybil attacks, while the overlay presented in \cite{kuhn2010towards} is resilient to an adversary that can force several peers to crash and join in a arbitrary manner.   In \cite{kuhn2010towards}, the number of join and leave operations tolerated at each turn is proportional to the degree of the nodes, which is optimal. However, the communication cost for maintaining the overlay is high as all the nodes of the network exchange messages at each step.

Other protocols considered unstructured overlays.
The protocol of \cite{law2003} builds an overlay corresponding to an expander graph obtained from the union of several random cycles. This protocol has been further extended and analyzed in \cite{Aspnes08,amin2004}.  Maintaining unstructured overlays induces fewer message exchanges compared to structured overlays \cite{kuhn2010towards,ratnasamy2001scalable,rowstron2001pastry,zhao2004tapestry} since only a polylogarithmic number of nodes are involved in the communication upon a join or a leave operation.  
Some of the previous constructions \cite{Aspnes08,amin2004} and \cite{law2003} can be complemented by a recent protocol from Pandurangan and Trehan \cite{Pandurangan2011Xheal} which preserves the expansion properties of a graph upon adversarial node removals. Nevertheless, the healing procedure proposed does not ensure an absolute expansion factor as we do.

The closest to ours, from the model perspective (dynamic network), is the one developed by Awerbuch and Scheideler \cite{awerbuch2004group,scheideler,Awerbuch2009,Scheideler2005}. They consider a synchronous network in which an adversary can force nodes to join and leave at each time step, with the constraint that the number of nodes in the network is always within a constant factor of the initial size. Their protocols further require that initially the network is exclusively composed of honest nodes and that the Byzantine ones join the network only after a particular initialization phase has taken place. Within this model, the authors propose a technique to maintain clusters of size $O(\log n)$ composed of a majority of honest ones. Our approach improves upon these previous works in several ways as we do not assume that initially the network is exclusively composed of honest nodes, we describe more precisely how to distributively perform all the operations, and, more importantly, we maintain a partition of the nodes when the size of the network varies polynomially. 

\section{Concluding Remarks}\label{conclusion}

This paper answers positively the following question raised in
\cite{King2010}:
``Can we [..] address problems of robustness in networks subject to
churn? An idea is to assume that:
1) the number of processors fluctuates between $n$ and $\sqrt n$
where $n$ is the size of name space; 2) the processors do not know
explicitly who is
in the system at any time; and 3) that the number of bad processors in
the system is always less
than a 1/3 fraction. In such a model, can we 1) do Byzantine agreement;
and 2) maintain small
(i.e. polylogorathimic size) quorums of mostly good processors?''

Our clustering protocol  can be leveraged to implement efficient and robust
algorithms for various problems such as broadcast, agreement,
aggregation, and sampling  in the context of highly dynamic networks. A
broadcast algorithm  using our technique would have for instance $\tilde
O (n)$ message complexity as compared to $O(n^2)$ without the
clustering. Similarly, a sampling algorithm relying on our protocol
would have a $\mathrm{polylog}(n)$ message complexity per sample. (We
discuss these and other applications
in the appendix).

We currently seek schemes to 
alleviate the need of the assumption of synchronous nodes. 
Another objective is to devise a procedure for the initialization phase of NOW whose communication cost is $o(n_{t_0}^2)$ (as opposed to $O(n_{t_0}^3)$ currently).
\newpage

\bibliographystyle{abbrv}
\bibliography{Long}

\appendix

\section*{Appendix}

\section{Continuous Time Random Walk} \label{randomWalk}

We briefly review fundamental results on continuous time random walk (CTRW)~\cite{BookMarkovProcesses}, which   we use as a building block for both our protocols OVER and NOW.
 
Given an undirected graph $G=(V,E)$, a CTRW is described by the following stochastic process: a virtual agent walks from node to node through edges chosen uniformly at random from the ones incident to the node on which the agent currently is. The walk is scheduled for a given amount of time $T$, and when the agent visits a given node $v$, it decrements a counter, representing the remaining time of the random walk,  by $\log(1/U)/d_{v}$, where $U$ is a number chosen uniformly at random from $(0,1)$. As long as the value of the counter is positive, the agent chooses at random a new neighbor and walks to this node. Otherwise, the walk stops. We denote by $\psi_t(u)$ the probability vector of the position of the agent at time $t$ for $u \in V$ the starting node of the CTRW. 

This type of CTRW has a uniform stationary distribution $\pi=(1/n)_i$ \cite{BookMarkovProcesses}, and its speed of convergence is characterized by the mixing time. More precisely, for every $\epsilon >0$, the $\epsilon$-mixing time of a CTRW is $T_{mix}(\epsilon)=\max_{u \in V}\min\{t|d(\psi_{t'}(u),\pi)\leq \epsilon, \forall t'>t\}$, where $d(.,.)$ is a distance function such as the maximum absolute difference of coordinates between two vectors.

The interpretation of \cite{Lindvall2002} (Theorem 5.2) states that $1/\epsilon$ represents the expected number of samples needed before retrieving an improperly selected node compared to the stationary distribution $\pi$. As we rely on random walks to generate samples, which is done a number of times polynomial in $n$, we use $\epsilon = \Theta(1/n^{c})$ for a chosen constant $c$. The choice of this value for $\epsilon$ means that whp all our samples can be considered as being picked uniformly at random. $T_{mix}(\epsilon)$ can be upper bounded by using $\lambda_2$ ($d(\psi_t(u),\pi)\leq \frac{\sqrt{n}}{2}e^{-\lambda_2t}$, \cite{DC07}), the second eigenvalue of the Laplacian matrix of $G$, which itself can be lower bounded using $I(G)$ (the isoperimetric constant of $G$ defined as $I(G)= \inf_{S:|S|\leq n/2}E(S,\bar S)/|S|$ where $E(S,\bar S)$ is the number of edges between $S$ and $\bar S= V\setminus S$), and $\Delta$, its maximum degree ($\lambda_2 \geq I(G)^2/2\Delta(G)$,  \cite{krebs2011expander}). For graphs from the Erd\"os-R\'enyi ${\cal{G}}(n,p)$ model \cite{Bollobas98}, for an arbitrarily small positive constant $\alpha$, $p = \log(n)^{1+\alpha}/n$ and $d=np$, whp $\lambda_2  \geq d^2/8\Delta(G)$ (Theorem~5.4 of \cite{17}).

If the nodes of the graph are weighted by a weight function $w: V\rightarrow (min,max)$, to bias the CTRW towards those with a bigger weight, once the CTRW is over, finishing at node $v$, a random value between 0 and 1 is chosen. If it is smaller than $w(v)/max$, $v$ is returned, otherwise a new CTRW is started. When $min/max$ is constant, the expected number of computed CTRW is constant and of $O(\log n)$ whp. The stationary distribution becomes $\left(\frac{w(v)}{\sum w(v)}\right)$.

\section{OVER: Expander Graphs}\label{expander}

OVER (for Over-Valued Erd\"os-R\'enyi graph) maintains an overlay modeled by a graph whose vertices are the clusters forming the partition of the nodes of the network. The term \textit{vertex} is used in the context of the overlay and  \textit{node} is used in the context of the network. Since each cluster  contains more than two thirds of honest nodes whp, we can assume in this section that all the vertices of the overlay are honest.  We  use the notation $\G^R$ where the $\hat\ $ relates to the fact that we consider an overlay, and the $\ ^R$ that it is an instance of a random graph. The evolution of the graph is represented by a sequence $\G^R_{t_0},\ldots, \G^R_{t_i}, \ldots$. $n_{t_i}$ denotes the number of vertices of $\G^R_{t_i}$. 
OVER relies mainly on four subroutines: {\CTRW}, {\join}, {\leave} and {\link} (detailed below). These are devised in a such a way that, starting from a random graph drawn from the Erd\"os-R\'enyi model, whp, at any time during a sequence of vertex additions and removals polynomial in $N$, the resulting graph has a large isoperimetric constant and a low degree (Theorems \ref{thmI(G)} and \ref{thmMaxDegree}). Furthermore, as we will  explain, this graph is robust against random vertices removal performed without calling the subroutine {\leave}\footnote{When a vertex of the overlay is removed without calling {\leave}, this means that the number of active nodes of the corresponding cluster $C$ is not sufficient to send valid messages, i.e. that more than half of the nodes of $C$ have crashed, and hence no valid message can be send by more than half of the nodes originally in $C$.} at the end of Section~\ref{AnalysisOVER}. 

\subsection{OVER: basic primitives}
\begin{itemize}[noitemsep,nolistsep]
\item ${\CTRW}(v)$ returns a vertex chosen by a CTRW (Continuous Time Random Walk, \emph{cf.} Appendix \ref{randomWalk}) from vertex $v$. The communication cost and round complexity of this subroutine are twice the length of the path performed by the CTRW, which is $O(\log^{2} N)$ as proven in the next subsection.
\item ${\link}(u,v)$ adds an extra edge between the vertices $u$ and $v$. The communication cost and round complexity of this subroutine are the length of the path used to communicate between $u$ and $v$.
\item ${\join}(v)$ is executed by a vertex $v$ contacted by a vertex $u$ upon joining the network. $2\log^{1+\alpha} N$ edges are added at random to connect $u$ to the rest of the graph using ${\link}(u,\CTRW(v))$.  ($\alpha$ is chosen positive constant that can be arbitrarily close to 0.)
\item ${\leave}(v)$ is executed by a vertex $v$ leaving the network without crashing. The edges connected to $v$ are removed and $2\log^{1+\alpha} N$ new edges are added at random using ${\link}({\CTRW}(v),{\CTRW}(v))$\footnote{Upon removal, adding edges is fundamental as otherwise the number of remaining edges in the graph may not be sufficient to guarantee connectivity after a number of vertex removals that is polynomial in $N$.}.
\end{itemize}
\begin{algorithm}[htb]
\caption{Continuous Time Random Walk: $\CTRW(v)$.}
\label{algo_id_ctrw}
\footnotesize
\begin{algorithmic}
\REQUIRE A connected graph $G=(V,E)$ whose Laplacian second eigenvalue is $\lambda_2$ and a starting vertex $v$.
\ENSURE The returned vertex is chosen uniformly at random.
\STATE $v$ sets $T=\log^2 n/\lambda_2$.
\STATE $v$ chooses at random a neighbor $u$ and moves to $current\_node=u$.
\WHILE{$T>0$} 
\STATE $current\_node$ chooses at random a number $U \in (0,1)$.
\STATE $current\_node$ updates $T=T-\log(1/U)/d_{current}$.
\STATE $current\_node$ chooses at random a neighbor $u$ and moves to $current\_node=u$.
\ENDWHILE
\STATE Return $current$ to the original vertex $v$ by following in a backward manner the path constructed by the CTRW.
\end{algorithmic}
\end{algorithm}
\begin{algorithm}[htb]
\caption{Adding a new edge: $\link(u,v)$.}
\label{algo_id_link}
\footnotesize
\begin{algorithmic}
\REQUIRE A connected graph $G=(V,E)$ and two vertices $u$ and $v$.
\ENSURE The addition of an edge between $u$ and $v$.
\STATE $u$ adds $v$ to its list of neighbors.
\STATE $v$ adds $u$ to its list of neighbors.
\end{algorithmic}
\end{algorithm}
\begin{algorithm}[htb]
\caption{Adding a vertex: $\join(v)$.}
\label{algo_id_add}
\footnotesize
\begin{algorithmic}
\REQUIRE A connected graph $G=(V,E)$, a new vertex $u$ that contacts a vertex $v$ already present in the graph.
\ENSURE The addition of $2\log^{1+\alpha} n$ edges at random that connect $u$ to the rest of the graph.
\FOR{$i=0; i=i+1; i<2\log^{1+\alpha} n$}
\STATE $v$ executes ${\link}(u,\CTRW(v))$.
\ENDFOR
\end{algorithmic}
\end{algorithm}
\begin{algorithm}[htb]
\caption{Removing a vertex: $\leave(v)$.}
\label{algo_id_remove}
\footnotesize
\begin{algorithmic}
\REQUIRE A connected graph $G=(V,E)$ and a vertex $v$ that has left $G$ in a proper manner (i.e, without crashing).
\ENSURE The addition of $2\log^{1+\alpha} n$ edges at random.
\FOR{$i=0; i=i+1; i<2\log^{1+\alpha} n$}
\STATE $v$ executes ${\link}(\CTRW(v),\CTRW(v))$.
\ENDFOR
\end{algorithmic}
\end{algorithm}

Note that the complexities given in this section will be increased when used in NOW (later), as each vertex of the overlay corresponds to several nodes, and a message from a vertex $C_1$ to another vertex $C_2$ induces $|C_1| \times |C_2|$ messages between nodes.

\subsection{Analysis of the OVER graph}
\label{AnalysisOVER}
We show here that at each time step, the graph constructed by OVER exhibits good expansion properties and a small maximum degree. These results are proved under the assumption that the random choices made during the construction of $\G^R$ are perfectly uniform (i.e, the small bias induced by the random walk is ignored). This assumption is justified by the fact that we consider a mixing time after which the distance from the distribution of the sample to the uniform distribution is $O(n^{-c})$ for some arbitrarily large constant $c$. 


We further demonstrate the results for a single addition or removal of vertex at once, but it can be easily extended to a higher number of additions and removals that could be performed in parallel (Section \ref{Weakening}).
\setcounter{theorem}{0}
\begin{theorem}[Isoperimetric constant of $\G^R$]\label{thmI(G)} 
Whp, at any time $t$ after a number of time steps polynomial in $n$, $\G^R_{t}$ has an isoperimetric constant $I(\G^R_{t}) \geq \log^{1+\alpha} N/2$.\end{theorem}
\begin{proof}
To prove this theorem, we demonstrate that at each time step $t_i$, $\G^R_{t_i}$ can be seen as an instance of a graph from the Erd\"os-R\'enyi model ${\cal G}(n_{t_i},p(n_{t_i}))$ with $p=\log(N)^{1+\alpha}/n_{t_i}$ to which some edges have been added.

If $p(n)$ is decreasing (i.e., $p(n+1)<p(n)$), a graph generated from the model ${\cal G}(n+1,p(n+1))$ can be considered as a sub-graph of a graph of ${\cal G}(n,p(n))$ to which a new vertex $v$ has been added and such that each new potential edge is created with probability $p(n+1)$. By drawing on this analogy, one can proceed as follow: first choose a degree $d_v$ for $v$ according to the binomial distribution $Bi(n+1,p(n+1))$, and then choose $d_v$ neighbors uniformly at random. We follow this procedure when we add a new vertex to $\G^R_{t_i}$ (join operation), with $p(n_{t_i})=(\log^{1+\alpha} N)/n_{t_i}$. The added vertex has a degree equals to $2\log^{1+\alpha} N$, which whp leads to a larger degree than $Bi(n_{t_i},\log^{1+\alpha} N/n_{t_i})$.

Similarly, when $p(n)=(\log^{1+\alpha} N)/n$, a graph issued from the model ${\cal G}(n,p(n))$ can be seen as a sub-graph of a graph of ${\cal G}(n,p(n+1))$ to which less than $2\log^{1+\alpha} N$ edges have been added at random. We follow this procedure when a vertex of $\G^R_{t_i}$ is removed (leave operation). Therefore, $\G^R_{t_i}$ can be seen as an instance of ${\cal G}(n_{t_i},p(n_{t_i}))$ to which some edges have been added. From \cite{17} and as $p(n_{t_i})n_{t_i} >> \log(n_{t_i})$, we 
have $I(\G^R_{t_i})\geq p(n_{t_i})n_{t_i}/2=(\log^{1+\alpha} N)/2$.
\end{proof}

To illustrate the previous lemma, consider the following case: starting from a Erd\"os-R\'enyi random graph with $\sqrt N$ nodes, we add nodes until we reaches $N$ nodes. The first nodes are connected by edges, each present with probability $O(\log^{1+\alpha} N /\sqrt N)$. When nodes are added, the probability of presence of the new edges will decrease, and it will be of $O(\log^{1+\alpha} N / N)$ for the last one. Therefore, we can see that the graph has a bigger density of edges among the initial nodes. This is why the obtained graph is not an instance of ${\cal G}(N,p(N))$, but rather such an instance with extra edges in between the oldest nodes. 

\begin{theorem}[Maximum degree of $\G^R$] \label{thmMaxDegree}
Whp, at any time $t$ after a number of time steps polynomial in $n$, $\G^R_{t}$ has maximal degree at most $c\log^{1+\alpha} N$ for some sufficiently large constant $c$. 
\end{theorem}

\begin{proof}
Given a sequence of graphs of the form $\G^R_{t_i}$, we want to compute the sequence of degrees of a specific vertex $v$. Let $t_{join}$ be the time at which $v$ joins the network. If $t_{join}=t_0$, then $v$ is in $\G^R_{0}$ the initial graph. Otherwise if $t_{join}>t_0$ then $v$ belongs to $\G^R_{t_{join}}$ but not to $\G^R_{t_{join-1}}$. We now focus on a sequence during which $v$ does not leave the network. During this sequence, the addition and removal of vertices have the following impact on the degree of $v$, for $n_{t_i}$ the number of vertices in $\G^R_{t_i}$ before the action performed at step ${t_{i+1}}$ is executed:
\begin{itemize}
\item When a vertex is added, $n_{{t_i}+1}=n_{t_i}+1$, and the degree of $v$ increases by one with probability $(2\log^{1+\alpha} N)/n_{t_i}$.
\item When a vertex is removed, $n_{{t_i}+1}=n_{t_i}-1$, and the degree of $v$ decreases by one with probability $d_t(v)/(n_{t_i}-1)$\footnote{the vertex which is removed is chosen among $(n_{t_i}-1)$ vertices as we work conditionally on the fact that $v$ is not removed.} as the vertex removed at random may be connected to $v$. Afterwards, as $2\log^{1+\alpha} N$ edges are added at random, the degree of $v$ increases by at most a value corresponding to the hyper-geometric distribution as $2\log^{1+\alpha} N$ trials are performed to 
select $n_{t_i}-2$ edges 
among $n_{t_i}-1 \choose 2$ possible edges in total. 
\end{itemize}

By assumption, the number of vertices in $\G^R_{t_i}$ at a particular time step $t_i$ verifies $\sqrt N \leq n_{t_i} \leq  N$. If we consider a sequence of addition and removal of vertices starting from $n_{t_0}=\sqrt N$, \footnote{Starting from a larger size would result only in adding less edges, therefore the degree would have less probability to go over $c\log^{1+\alpha} N$.} then there are at most $ N-\sqrt N$ more addition than removal operations. Moreover, each removal occurring at time $t_j$ can be associated to an addition that has occurred at time $t_i<t_j$ such that $n_{t_i}=n_{t_j}$.

Considering the event $\{d_t(v) \geq c\log^{1+\alpha} N\}$, we want to prove that its probability of occurrence is very low. Such an event would be preceded by another event:

 $\{d_{t'}(v) \geq \frac{c}{2}\log^{1+\alpha} N\}$ such that from $t'$ to $t$ the degree of $v$ remains higher than $\frac{c}{2}\log^{1+\alpha} N$. The probability that such an event occurs can be upper bounded by the probability of the following random variable being larger than $c\log^{1+\alpha} N$. 

For all time steps $t_i$, we have $\sqrt N \leq n_{t_i} \leq  N$ and we define $X=\sum_{j=\sqrt N}^N W_i + \sum_{i=0}^T (X_{i}+Y_{i} + Z_{i})$ in which: 
\begin{itemize}
\item $W_j=+1$ with probability $\log^{1+\alpha} N/j$.
\item $X_i=+1$ with probability $\log^{1+\alpha} N/n_{t_i}$.
\item $Y_i=-1$ with probability $\frac{c}{2}\log^{1+\alpha} N/n_{t_i}$ (as $d(v) \geq \frac{c}{2}\log^{1+\alpha} N$, we can lower bound the probability that it decreases).
\item $Z_i$ follows the hyper-geometric distribution corresponding to $2\log^{1+\alpha} N$ trials to select $n_{t_i}-2$ elements among $n_{t_i} -1 \choose 2$. 
\end{itemize}

Using standard Chernoff bounds, we obtain that in order to reach $\sum_{i=0}^T X_i \geq c\log^{1+\alpha} N$, $T$ needs to be sufficiently large with respect to $N$ so that $\sum 1/n_{t_i}$ is large enough. From this, it is possible to infer that the probability of the event $\{X>c\log^{1+\alpha} N\}$ is $o(n^{-c'})$ small in $N$ for an arbitrarily large $c'$ (depending on $c$). Therefore, the maximum degree of $\G^R_{t_i}$ is upper bounded whp by $c\log^{1+\alpha} N$ at each time step during a polynomial number of join and leave operations.
\end{proof}


The fact that more than half of the nodes of a cluster simultaneously leave the network (voluntarily or not)   corresponds to a crash of the corresponding vertex in the overlay. Effectively,  the cluster can then  no longer  send valid messages as a message  needs to be sent by at least half of the nodes of the cluster to be valid. The graph $\G^R_{t}$ obtained with OVER is robust against $\epsilon n_{t}$ random such crashes of vertices: the properties of Theorems~\ref{thmI(G)} and \ref{thmMaxDegree} will still be ensured whp. This resilience can be deduced from the proof of Theorem~\ref{thmI(G)} (which appears in the Appendix). 
(In comparison, a graph obtained using techniques presented in \cite{Aspnes08,amin2004,law2003} is composed of an union of cycles. 
Therefore if a vertex crashes, the cycles are cut and the protocol no longer works.)

\section{NOW: Initialization phase}

\subsection{Distributed random number generation} 
\label{randnum}
{\randNum} enables the nodes of a cluster to agree on a common integer chosen uniformly at random from the interval $(0,r)$ (the value of $r$ is determined by the protocol using \randNum. Here, $r$ represents a knowledge common to all the nodes involved.). In short, the idea is that each node generates a value uniformly at random from the interval $(0;r)$, and forwards it to the other nodes. The sum of the chosen values modulo $r$ gives the output. An unbiased random number can be securely computed with the protocol of \cite{BenOr1994} when private channels are available. Alternatively, one can use the technique proposed in \cite{RNG} which does not require private channels but tolerates only a sixth of Byzantine nodes and has little bias. The communication cost of these algorithms are $\mathrm{polylog}(N)$ and $O(\log^2 N)$ (we use  $O(\log^2 N)$ in the following, this allows us to keep track of the precise complexity rather than to have a $\mathrm{polylog}(N)$ one, but it applies only to $\tau \leq 1/6-\epsilon$). The round complexity are $\mathrm{polylog}(N)$ and $O(\log N)$.


\begin{algorithm}[ht]
\caption{Distributed random number generation: \randNum.}
\label{algo_id_randNum}
\footnotesize
\begin{algorithmic}
\REQUIRE A cluster $C$ with a majority of honest nodes and an integer $r$.
\ENSURE The generation of an integer chosen at random from the interval $(0;r)$.
\STATE Each node of $C$ chooses an integer chosen uniformly at random from $(0;r)$ and encrypts it using its session public key.
\STATE The node compute the sum modulo $r$ of the chosen number using the protocol of \cite{BenOr1994}. 
\STATE Output the obtained value.
\end{algorithmic}
\end{algorithm}

\subsection{Random cluster choice}

\begin{algorithm}[ht]
\caption{Randomly choosing a cluster: \randCl.}
\label{algo_id_randCl}
\footnotesize
\begin{algorithmic}
\REQUIRE A graph connecting clusters each with a majority of honest nodes and an initial cluster $C$.
\ENSURE The choice of a cluster uniformly at random among all the clusters.
\STATE The nodes from $C$ choose an integer $i$ at random using {\randNum} with the integer $r=d_C$ and set $T=8\log n$.
\STATE The nodes from the current cluster $C$ initiate a CTRW by sending a message to all the nodes of the cluster $C'$ for $C$ the $i^{th}$ neighbor of cluster $C$ on the overlay. 
\STATE $C'$ becomes the current cluster.
\WHILE{$T>0$}
\STATE The nodes from the current cluster choose a number $U$ from $(0,1)$ using {\randNum}, and reduce $T$ by  $\log(1/U)/d$, for $d$ being the degree of the current cluster.
\STATE The nodes from the current cluster choose an integer $i$ from $(1,d)$ using {\randNum}.
\STATE The nodes from the current cluster $C$ send a message to all the nodes of the cluster $C'$, which is the $i^{th}$ neighbor on the overlay of cluster $C$. 
\STATE $C'$ becomes the current cluster.
\ENDWHILE
\STATE Output the identity of the current cluster in which the walks has ended.
\end{algorithmic}
\end{algorithm}

\subsection{Exchange of nodes}

The algorithm related to the exchange of nodes is described in Algorithm \ref{algo_id_exchange}.

\begin{algorithm}[htb]
\caption{Exchange of nodes: \exchange.}
\label{algo_id_exchange}
\footnotesize
\begin{algorithmic}
\REQUIRE A graph connecting clusters with a majority of honest nodes and a cluster $C$.
\ENSURE All the nodes of $C$ are exchanged with nodes chosen at random.
\FOR{nodes $x$ in $C$}
\STATE Choose a cluster $C_x$ using \randCl.
\STATE The nodes from $C_x$ choose an integer $i_x$ using {\randNum} with $r=|C_x|$, which corresponds to a node $y_x$ of $C_x$.
\ENDFOR
\FOR{nodes $x$ in $C$}
\STATE All the nodes from $C_x$ send a message to all the nodes of the neighboring clusters that $x$ replaces $y_x$.
\STATE  All the nodes from $C$ send a message to all the nodes of the neighboring clusters saying that $y_x$ replaces $x$.
\ENDFOR
\end{algorithmic}
\end{algorithm}

\subsection{Global knowledge computation}

The goal of this algorithm, used during the initialization phase of NOW is to ensure that all honest nodes of a cluster know the identities of all nodes.  This algorithm terminates after a number of communication rounds at most the diameter of the graph considering only the edges adjacent to at least one honest node. When the algorithm terminates, it is guaranteed that all honest nodes know the identities of all nodes in the network. Its communication cost is $O(n\times e)$ where $e=|E|$ (see Appendix for the theorem and details).

\begin{algorithm}[htb]
\caption{Global knowledge computation}
\label{algo_id_distribution}
\footnotesize
\begin{algorithmic}
\REQUIRE A graph $G=(V,E)$ in which honest nodes form a connected component and each Byzantine node is adjacent to an honest node. A node $v$ knows its neighbors $\Gamma_v = \{u : uv \in E\}$.
\ENSURE All honest nodes know $V$.
\STATE Each node $v$ do:
\STATE Set its request list to $\Gamma_v$, $V_v^0 = \Gamma_v$ and $i=1$.
\WHILE{$v$ request list is not empty} 
\STATE round $j$: $v$ asks to all $u$ in its request list for $\Gamma_u$. $v$ set its request list to empty.
\FOR{all requests received by $v$}
\STATE round $j+1$: $v$ sends $\Gamma_v$ to the node requesting it if it has not done so yet.
\ENDFOR
\STATE round $j+1$: $V_v^i$ = $V_v^{i-1}$
\FOR{all set of neighbors $\Gamma$ received by $v$}
\STATE round $j+1$: $V_v^i = V_v^i \cup \Gamma$.
\ENDFOR
\STATE round $j+1$: v set $i=i+1$ and its list to $V_v^i \setminus V_v^{i-1}$
\ENDWHILE
\STATE Set $V$ to $V_v^i$.
\end{algorithmic}
\end{algorithm}

\begin{theorem}[Global knowledge computation]
\label{theo_global_knowledge}
In a graph composed of $n$ vertices, Algorithm~\ref{algo_id_distribution} terminates after a number of communication rounds at most the diameter of the graph considering only the edges adjacent to at least one honest node. When the algorithm terminates, it is guaranteed that all honest nodes know the identities of all nodes in the network. Its communication cost is $O(n\times e)$ where $e=|E|$.
\end{theorem}
\begin{proof}
Each node sends its list of original neighbors to each other node exactly once, therefore its complexity is $\sum (n-1) \times d_v = O(n\times e)$. It further contact each node exactly one to ask for the set of neighbors. This does not change the asymptotic complexity.

We are especially interested in the graph $G'$ obtained when considering only the edges with at least one honest node. We denote by $d'(u,v)$ the distance between $u$ and $v$ in $G'$.

We prove the theorem by induction. The induction hypothesis is that at the end of the $i^{th}$ times $v$ went through the while loop, $V_v^i$ contains at least all the nodes at distance less than $i+1$ in $G'$.

Before $v$ enters the while loop, $V_v^0$ contains all the nodes at distance $1$ in $G'$. During the first execution of the while loop, $v$ receives at least all the neighbors of its honest neighbors, therefore the induction hypothesis is true for this first time as $V_v^0$ contains $\{u : d'(u,v) \leq 2\}$.

Now we consider an induction step. $V_v^{i-1}$ contains all the nodes at distance $i$ in $G'$. During the while loop, $v$ receives the neighbors (in $G$) of all honest nodes in  $V_v^{i-1}$ from which it has not heard yet. This ensures that it has received all the nodes at distance $i+1$ in $G'$, proving the induction step is valid.

The while loop terminates when the request list is empty. Such an event occurs when $v$ has $V_v^{i-1}=V_v^{i}$ which means that no nodes that $v$ contacted 
are 
nodes that were unknown to $v$. In other words, it means that all the neighbors of the nodes of $V_v^{i-1}$ are in $V_v^{i-1}$. But since the honest nodes form a connected component and that all the nodes controlled by the adversary are adjacent to at least one honest node, we conclude that $V_v^{i-1}=V$.

Finally, the induction hypothesis directly gives us that the algorithm stops after a number of communication rounds at most the diameter of $G'$.

This theorem requires that each node controlled by the adversary is adjacent to at least one honest node as otherwise, if a node $x$ does not satisfy this hypothesis, the adversary may decides that some honest node know about $x$ while some others do not. This may be problematic to apply the protocol presented in \cite{King2010}. On the other end, this hypothesis can be dropped if we use protocols that do not required each honest node to have the same view of the network.
\end{proof}

\subsection{Weakening the assumptions}\label{Weakening}
We discuss here how  to weaken some of the assumptions upon which NOW has been analyzed in order to increase the generality of our construction. In particular, we describe how to adapt NOW so that it can tolerate several parallel join and leave operations.

\paragraph{Occurrence of several join and leave operations at a particular time step}
In order to accommodate a high number of nodes joining and leaving at each time step, it is sufficient to consider clusters of larger size. For instance, to be able to cope with $\log^i N$ nodes joining or leaving the network at each time step for some constant $i\geq 0$, we can use clusters of size $\log^{i+1} N$ nodes instead of $\log N$. As a consequence, at each time step, the number of nodes susceptible to leave a cluster is negligible compared to its size. Hence, the adversary cannot control a majority of the nodes of a cluster by forcing all the nodes from this cluster to leave at the same time step. All the proofs that we have developed can be adapted to cope with this new cluster size. Moreover, with respect to the operations used by   OVER and NOW, they can all be performed in parallel and therefore no further adaptation is required.
Finally, the new construction does not impact the round complexity and increases only the communication cost by a polylogarithmic factor.

\paragraph{Multiple crashes}\label{multipleCrashes}
Our protocol tolerates an adversary that can make $\eta n$ random nodes crash simultaneously if we suppose that it controls at most a $1/3-\epsilon$ fraction of the nodes and that $\eta$ is small enough when compared to $\epsilon$. This type of crashes can be used to simulate a failure of some critical links in a network. 
The assumption that the crashes are random is necessary as, otherwise, the adversary could split the honest nodes into disconnected components and gain the lead in some clusters, which will make NOW (and any other protocol) fail.

If we assume that the adversary controls at most a $1/6-\epsilon$ fraction of the nodes, then regardless of the number of simultaneous honest nodes crashing, the adversary will not control more than a third of the nodes in a cluster whose remaining size is greater than half of its size before the crash.

\section{Applications}


The overlay of clusters we obtain (using NOW and OVER) can be applied to solve a wide range of  problems in distributed computing. We review some of these applications in this section. The layering of the corresponding  algorithms is depicted in Figure \ref{stack}. The complexities obtained by the different applications are summarised later. 
Figure \ref{stack} describes the interactions between the different algorithms. 


\begin{center}\label{stack}
\begin{figure}[tb!]
  \begin{center}
\scalebox{0.4}{\includegraphics{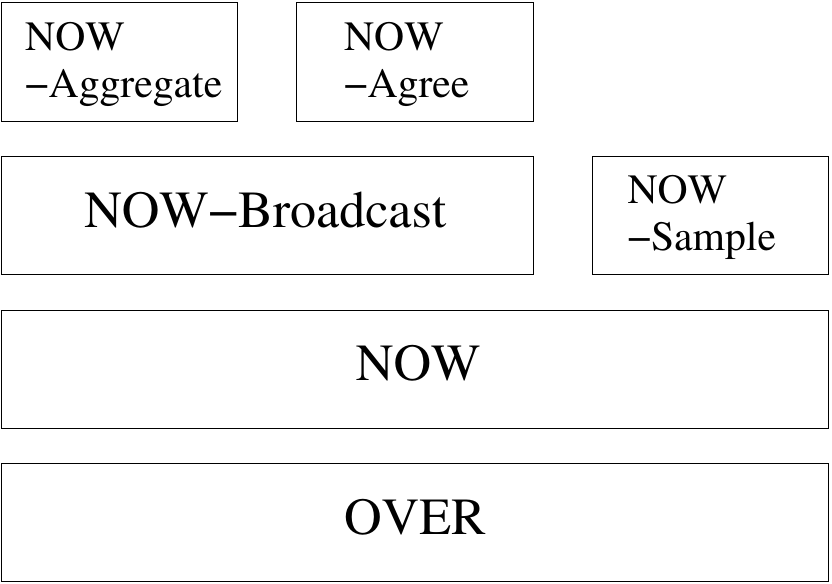}}
\caption{Overview of the algorithms stack.}
   \end{center}
\end{figure}
\end{center}

\subsection{NOW-Broadcast} 
\emph{Broadcast}, first introduced in \cite{pease1980reaching}, is one of the fundamental primitives in distributed computing. In a nutshell,  the sender, a specific node in the network, aims at sending a message to all nodes in the network such that (a) either they all   receive the same message,
or (b) the broadcast is aborted and none receive it 
(this could happen for instance if the sender crashes or tries to send different messages to different nodes). 

Broadcast is particularly challenging in a system in which the adversary controls some of the nodes.  Current broadcast algorithms that tolerate an adversary controlling a constant fraction of the nodes have a communication cost that is at least quadratic in the size of the network (compared to linear when there is no adversary). 
The design of a broadcast algorithm with a lower communication cost would increase the efficiency of any algorithm relying on a broadcast channel. However, such an algorithm  cannot be deterministic.
In the following, we discuss how to use our overlay of clusters to design a robust and efficient probabilistic broadcast algorithm.

\subsubsection{Broadcasting in a network composed only of honest nodes}
If there is no adversary and the network is a good expander graph, it is easy to design an efficient broadcast algorithm: each node simply forwards the message to all its neighbors. The time complexity equals the diameter  $\diameter$ of the graph (defined as the maximum over the length of the shortest path between two vertices), while the communication cost equals the number of edges. An even more efficient solution consists in sending the message only through the edges of a precomputed tree $\mathcal T$ of depth $\diameter$, which leads to a number of messages sent equals to the number of nodes (minus one). Such a tree $\mathcal T$ can be computed via a breadth first search for instance.

Relying on the overlay of clusters built and maintained by NOW and OVER, it is possible to derive a broadcast algorithm that can tolerate an adversary and whose communication cost is $\tilde O(n)$  
The algorithm is basically the following:
\begin{itemize}
\item When a node $u$ wants to broadcast a message, it first broadcasts it to all  other nodes in its cluster using a secure broadcast algorithm (e.g., \cite{Hirt2010}
whose complexity is quadratic in the size of the cluster). 
\item A cluster receiving the message propagates it to all its neighbors in $\mathcal T$: each node forwards the message to all  nodes of the neighboring clusters iff the message is supported by a majority of nodes of the original cluster. 
\end{itemize}

The global communication cost of the algorithm is $O(n\log^{2+2\alpha} N)$ as there are $n-1$ edges in $\mathcal T$, each inducing an exchange of $O(\log^{2+2\alpha} N)$ messages.

\begin{algorithm}
\caption{NOW-Broadcast}
\label{algo_id_broadL}
\begin{algorithmic}
\REQUIRE A network with a partition and an overlay maintained by NOW.
\ENSURE The same message will be broadcast to all the nodes of the network whp.
\STATE The sender sends its message $\cal M$ using a secure broadcast algorithm to all  nodes in its cluster.
\STATE A node that has received $\cal M$ (after a majority vote) during the previous time step forwards it to all nodes of the neighboring clusters.  
\end{algorithmic}
\end{algorithm}

\subsection{NOW-Agree}

We now show how to leverage our clustering technique (i.e., overlay of clusters)
 to implement an efficient solution to the Byzantine Agreement (BA) problem \cite{Lamport1982}. 
In a network composed only of honest nodes, agreement can be easily solved by having one of the nodes in the network (e.g., the one with the smallest identifier) send its input to all the other nodes.
In the presence of an adversary controlling a fraction of the nodes and assuming that each node has a
global knowledge of the network, the algorithm proposed by King and Say~ \cite{King2010} solves the BA problem with a communication cost of $\tilde O(n\sqrt n)$. 
In our setting, assuming  NOW and OVER, 
a  cluster initiates a BA algorithm as follow:
\begin{itemize}
\item The nodes of this  cluster run a BA algorithm among themselves.
\item These nodes then  broadcast the result of the BA to the rest of the nodes of the network by calling the NOW-Broadcast algorithm described in the previous subsection.
\end{itemize}

When the identity of the cluster initiating the agreement algorithm or the cluster with the smallest identifier is not clear, we need a procedure to initiate the agreement algorithm on at least one cluster, but  not too many ones.
By assumption, we know that the size of the network is between $\sqrt N$ and $N$. Therefore, if each node initiates the agreement algorithm with probability $\log N/N$, and if after $\log^4 N$ steps no output is received, it means that no node has initiated the algorithm. In this case, each cluster proceeds by initiating the algorithm with a probability that is twice the previous one. One can show that this procedure has to be repeated at most $\log N$ times in order to ensure that each cluster receives at least one output whp. In this case, $O(\log N)$ clusters will have broadcast a message, which results in the communication cost being increased by a factor $\log N$ as compared to the original broadcast algorithm. In order for a node to choose among the multiple outputs it receives, a cluster $C$ broadcasting a message attaches to it a tag that corresponds to the lowest id   within the cluster $C$. Therefore, the final output selected by a node is the message received whose attached id is the lowest.

\begin{algorithm}
\caption{NOW-Agree}
\label{algo_id_agree}
\begin{algorithmic}
\REQUIRE A network with a partition and an overlay maintained by NOW. All nodes have an input bit. A node $u\in C$ initiates the protocol NOW-Agree. 
\ENSURE All  honest nodes agree on a bit that was proposed initially by one of them.
\STATE Nodes in $C$ run a BA algorithm among themselves (such as \cite{King2010}) and output $b$.
\STATE Cluster $C$ broadcasts $b$ to the network by using NOW-Broadcast.
\end{algorithmic}
\end{algorithm}

\subsection{NOW-Aggregate} 
In \cite{SRDS12}, an algorithm has been proposed to compute aggregate functions in a secure and scalable manner by relying on a ring overlay (i.e., an overlay in which the clusters are organized in a ring). This type of overlay can either be maintained with $\G^R$ or computed from scratch when needed. Alternatively, one can construct a tree via a breadth-first search started on a chosen cluster (for instance via a BA algorithm such as the one described in the previous subsection). This later option produces a structure of small diameter and therefore can be used to  improve the round complexity of the aggregation algorithm. 
Using the algorithm of \cite{SRDS12} leads to a communication cost of $O(n\log^2 N)$.

If privacy is not a primary concern, in the sense that the adversary can learn information about the input of honest nodes, one can rely on the algorithm proposed in \cite{MoskSeparable} to compute efficiently an aggregation function. Thereafter, we describe the pseudo code of such an algorithm in the situation in which the aggregation function considered is the sum. However, it is  straightforward to adapt it to any separable function as defined in \cite{MoskSeparable}. We refer the reader to the original paper for the description of the full algorithm and its detailed analysis.

\begin{algorithm}
\caption{NOW-Aggregate}
\label{algo_id_aggregate}
\begin{algorithmic}
\REQUIRE A network with a partition and an overlay maintained by NOW. Each node $u$ has a positive integer $y_u$ as input and knows a common integer $r$ measuring the accuracy of the estimated value obtained.
\ENSURE Each node learns an estimate of the sum of all the values, i.e. $\sum_{v\in V} y_v$.
\STATE For each cluster $C$, the nodes of $C$ broadcast their input to all the other nodes of $C$.
\STATE Each node of $C$ computes $y_{C}=\sum_{v\in C} y_v$.
\STATE Nodes of $C$ collaboratively generate $r$ independent random numbers $W_1^C,\dots W_r^C$ using the primitive $\randNum$, such that the distribution of each $W_l^C$ is exponential with rate $y_C$ for $l = 1,\dots,r$.
\STATE Each cluster broadcasts $W_1^C,\dots W_r^C$ using NOW-Broadcast.
\STATE Each node computes $\tilde W_l=\min_C W^C_l$ for $l = 1,\dots,r$. 
\STATE Each node $u$ computes $\frac{r}{\sum_{l=1}^r \tilde W_l}$, which is output as the estimate of $\sum_{v\in V} y_v$.
\end{algorithmic}
\end{algorithm}

The propagation of the minimum of each $W_i^C$, for $1\leq i \leq r$, leads to a communication cost for the broadcast algorithm of $\tilde O(nr\diameter)$.

\subsection{NOW-Sample} 

A peer sampling service provides each node with a sample of nodes picked uniformly at random~\cite{jelasity2004peer,BaldoniPQS10}. 
We can derive a  NOW-Sample algorithm  as follows: first, a node calls the peer sampling service at the level of its cluster; afterwards, its cluster initiates a biased CTRW to select a cluster. The chosen cluster then designates a node chosen uniformly at random, whose identifier is sent back to the node requiring the sample. This algorithm has a polylogarithmic round complexity and is repeated several times to obtain a sample of the desired size. 

\begin{algorithm}
\caption{NOW-Sample}
\begin{algorithmic}
\REQUIRE A network with a partition and an overlay maintained by NOW, as well as two parameters $\delta \in (0,1)$ and a node $u$ requiring a sample.
\ENSURE The node $u$ receives the id of a node chosen uniformly at random.
\STATE  The node $u \in C$ broadcasts a message to all the nodes of its cluster requiring a sample. 
\STATE $run = true$
\WHILE{$run$}
\STATE The cluster $C$ initiates \randCl, which returns $C'$.
\STATE $C'$ set $run=false$ with probability $|C'|/kl\log n$ using \randNum. Otherwise set $C=C'$.
\ENDWHILE
\STATE The cluster $C'$ selects one of its nodes $v$ at random using the primitive \randNum.
\STATE The cluster $C'$ sends $v$ to $u$ following the path used by {\randCl} in a backward manner.
\end{algorithmic}
\end{algorithm}

\subsection{Summary}

The complexities obtained by the different applications are summarized below. 

\setlength{\extrarowheight}{3 pt}

\begin{center}
\begin{tabular}{|l|c|c|}
	\hline
	Protocol & Communication cost & Round complexity\\
	\hline
   	NOW-Broadcast & $\widetilde O(n)$ & $O(\diameter)$ \\
	\hline
	NOW-Agree & $\widetilde O(n)$ & $O(\diameter)$\\
	\hline
	NOW-Aggregate & $\widetilde O(nr\diameter)$ & $O(\diameter)$\\
	(standard version) &  & \\
	\hline
	NOW-Aggregate & $\widetilde O(n)$ &  $\widetilde O(n)$\\
	(privacy-preserving version)  &  & \\
	\hline
	NOW-Sample (per sample) & $\mathrm{polylog}(N)$ & $\mathrm{polylog}(N)$\\
	\hline
\end{tabular}
\end{center}

\end{document}